\documentclass[11pt,letter,citesort]{article}

\usepackage[margin=.9in]{geometry}

\usepackage[utf8]{inputenc}
\usepackage{amsmath}
\usepackage{amssymb}
\usepackage{amsfonts}
\usepackage{amsthm}
\usepackage{enumitem}
\usepackage{mathtools,bbm}
\usepackage{amscd}
\usepackage{textcomp}

\usepackage{latexsym}
\usepackage{wasysym}
\usepackage{graphicx}
\usepackage[table,pdftex,hyperref,svgnames]{xcolor}
\usepackage{hyperref,doi}
\hypersetup{colorlinks=true, linkcolor=blue, breaklinks=true, urlcolor=blue}

\usepackage[all,cmtip]{xy}
\usepackage{bbold}

 \usepackage{lipsum}
\usepackage{mathtools}
 \usepackage{cuted}

        \usepackage{setspace}

\newcommand\blfootnote[1]{%
  \begingroup
  \renewcommand\thefootnote{}\footnote{#1}%
  \addtocounter{footnote}{-1}%
  \endgroup
}

\usepackage[caption=false,font=footnotesize]{subfig}
\graphicspath{{./images/}{./fig/}}

\newcommand{\N}{\mathbb{N}}
\newcommand{\Z}{\mathbb{Z}}

\newcommand{\E}{\mathbb{E}}

\newcommand{\mcE}{\mathcal{E}}
\newcommand{\prob}{\mathbb{P}}

\newcommand{\until}[1]{\{1,\dots, #1\}}

\newtheorem{theorem}{Theorem}[section]

\newtheorem{lemma}[theorem]{Lemma}

\newtheorem{remark}[theorem]{Remark} 
{
      \theoremstyle{plain}
            
}
\newtheorem{definition}[theorem]{Definition}

\DeclareSymbolFont{bbold}{U}{bbold}{m}{n}
\DeclareSymbolFontAlphabet{\mathbbold}{bbold}
\newcommand{\vect}[1]{\mathbbold{#1}}

\newcommand{\setdef}[2]{\{#1 \; | \; #2\}}

\usepackage{algpseudocode,algorithm,algorithmicx}

\setlist{nosep}

\newcommand{\union}{\operatorname{\cup}}

\def\BibTeX{{\rm B\kern-.05em{\sc i\kern-.025em b}\kern-.08em
    T\kern-.1667em\lower.7ex\hbox{E}\kern-.125emX}}

\begin{document}
\title{Polarization and Fluctuations \\ in Signed Social Networks} 

\author{Pedro Cisneros-Velarde, Kevin S. Chan, Francesco Bullo}
%
%
%

\maketitle

\begin{abstract}
   Much recent research on social networks has focused on the modeling and
   analysis of how opinions evolve as a function of interpersonal
   relationships.  It is also of great interest to model and understand the
   implications of friendly and antagonistic relationships.  In this paper,
   we propose a new, simple and intuitive model that incorporates the
   socio-psychological phenomenon of the boomerang effect in opinion
   dynamics. We establish that, under certain conditions on the structure
   of the signed network that corresponds to the so-called~\emph{structural
     balance} property, the opinions in the network polarize.  Compared to
   other models in the literature, our model displays a richer and perhaps
   more intuitive behavior of the opinions when the social network does not
   satisfy structural balance. In particular, we analyze signed networks in
   which the opinions show persistent fluctuations (including the case of
   the so-called \emph{clustering balance}).
\end{abstract}

\section{Introduction}
\blfootnote{This work was supported by, or in part by, the Army Research
    Laboratory and the Army Research Office under grant number
    W911NF-15-1-0577, the Cooperative Agreement Number W911NF-18-2-0066,
    and the National Science Foundation under grant number DGE-1258507. The
    views and conclusions contained in this document are those of the
    authors and should not be interpreted as representing the official
    policies, either expressed or implied, of the Army Research Laboratory
    or the U.S. Government. The U.S. Government is authorized to reproduce
    and distribute reprints for Government purposes notwithstanding any
    copyright notation herein.}
\blfootnote{Pedro Cisneros-Velarde (e-mail:
    pacisne@gmail.com) and Francesco Bullo (e-mail:
    bullo@engineering.ucsb.edu) are with the Center for Control, Dynamical
    Systems and Computation, University of California, Santa Barbara.}
\blfootnote{Kevin S. Chan (e-mail: kevin.s.chan.civ@mail.mil) is with the U.S.\ Army Research Laboratory, 2800 Powder Mill Rd., Adelphi, MD 20783, USA.}
There have been various opinion dynamics models in the literature
\cite{DA-AO:11,AVP-RT:17}. Opinions can be modeled as real numbers taking
values in the closed interval $[0,1]$, where $0$ means an agent completely
disagrees with a particular issue, and $1$ that she completely agrees. One
important question to answer is how the evolution and final distribution of
opinions in a social network depend on the underlying network's topology
and of the (positive) influence structure among the individuals. More
recently, signed graphs were introduced into the opinion dynamics
literature. Signed graphs represent a natural way to model positive and
negative relationships among individuals. For example, a sociological
relevant concept is \textit{structural balance}, in which the members of a
social network can either have only positive relationships or be divided in
two factions in which members of the same faction have positive
relationships but negative ones with members of the other faction. The
seminal work by Altafini \cite{CA:13} proposed a continuous time model over
a signed graph where the opinions can take any real value. It is shown that
when the underlying graph satisfies structural balance, the opinions
converge to bipartite consensus and polarize, i.e., all opinions have the
same absolute value with their signs indicating which agents belong to the
same faction (if there is one faction, all opinions have the same sign).  A
discrete-time signed opinion model which is a counterpart of the Altafini
model has also been proposed~\cite{JMH:14,ZM-GS-KHJ-MC-YH:16}, in which
bipartite consensus is also attained under structural balance. These two
models have initiated a lot of research in the field of signed opinion
dynamics, and are, arguably, the most popular models in the literature.
Extensions of these models and further analysis have been done in the
literature, as can be noted in the recent work~\cite{JL-XC-TB-MAB:17} and
the references therein.  Note, however, that both Altafini models and their
extensions present an unrealistic opinion vanishing behavior (i.e., the
opinions converge to zero) whenever the property of structural balance is
lost in the underlying social network.

Another class of models in opinions dynamics was proposed by Li et
al. \cite{YL-WC-YW-ZZ:15} and is based on an extension of the voter model
to signed graphs. In this model, individuals initially take binary opinion
values (e.g., $0$ and $1$). Then, at each subsequent time step, an
individual is selected according to some process and updates her opinions
by copying the same or the opposite opinion of one of her neighbors
according to the sign of their relationship.  By design, opinions
cannot vanish under generic signed networks; however, the opinion values
are simply discrete. Whenever the graph satisfies structural balance, they
showed that the opinions polarize: one faction takes one value, while the
other faction takes the remaining one. Recently, Lin et
al. \cite{XL-QJ-LW:18} proposed a model which can be regarded as an
extension to the one from Li et al. In this model, opinions can take $m$
different discrete values from a set $S$. Then, an individual will copy the
same opinion from a positive neighbor, but when facing a negative one, will
randomly select an opinion different from that neighbor from the set $S$.


In this paper, we propose a novel opinion model over signed graphs. We
assume that the opinions are real numbers taking value in a closed interval
and each edge of the graph indicates the friendly or antagonistic
relationship between two individuals.  Our model is inspired by the
\emph{boomerang effect} studied in social
psychology~\cite{ARC:62,SB-PSH:09,RPA-JCM:67}, which aims to explain why in
some situations where two individuals engage in communication, they may not
end up being in a better agreement but rather their attitudes become more
dissentive, i.e., their opinions do not go in the \emph{intended direction}
(e.g., consensus or agreement) but in the \emph{opposite direction} (e.g.,
polarization). The early work~\cite{CIH-OJH-MS:57} suggested that this
phenomenon can be explained by \emph{``the relative distance between
  subjects' attitudes and position of communication"}. Our model is
motivated by the empirical observations in the social sciences (e.g., from
the study of interpersonal attraction~\cite{NJS-PBB:01}) that two friendly
agents will be closer in their attitudes and perspectives than two
unfriendly agents.  Specifically, we make the following assumption:
whenever two agents who have a positive relationship interact, they are
more agreeable and their opinions will become closer or even be in
consensus, i.e., the opinion \emph{changes in the intended direction}. On
the other hand, whenever two agents with a negative relationship interact,
the differences in their opinions will be more polarized after the
interaction because of their increasing disagreement, i.e., the opinion
\emph{changes in the opposite direction}. Our opinion model captures such
behavior mathematically, and we call it the \emph{affine boomerang
  model}. Mathematically, our proposed model is an affine model, which
makes it remarkably simple, and its dynamics are self-explanatory. Besides
a linear model like the discrete Altafini model, this is, arguably, the
next simplest model structurally.

Our second contribution is a formal analysis of our proposed
  model: under certain conditions on the sign structures of the network
that corresponds to structural balance, our model expresses opinion
polarization, i.e., the opinions of two groups converge to opposite extreme
values of the closed interval.

Finally, it is important to compare our model and the aforementioned models
in the literature. Our model has the property that opinions do not
necessarily vanish whenever the graph is not balanced, but, for example,
can continue fluctuating inside the closed interval. The vanishing
behavior, which we mentioned happens in both types of Altafini models and
their extensions, has been interpreted as if the agents in the network
become neutral or indifferent towards a specific topic.  In the case of
three antagonistic groups in a connected network, this would mean that all
groups will end up having a zero valued opinion, i.e., they will have
consensus on \emph{not having an opinion}. This might be difficult to
interpret. Instead, our proposed model predicts that two groups will
polarize their opinions and the third one will continue fluctuating its
opinions since its members observe people they dislike having opposite
opinions. Thus, this third group does not settle down to a definite opinion
and its members are persistently disagreeing with each
  other. This is, arguably, more intuitive since individuals of a social
network can always hold an opinion, independently of whether their network
is balanced or not. Moreover if we have an unbalanced network that differs
from a balanced one in just the sign of one edge, it is not clear why that
would drive the whole social network towards an indifferent
opinion. Instead, our model suggests that opinions may fluctuate around
extreme values of opinion, which is more intuitive since the underlying
social network is \emph{approximately} balanced.
%
%
%

\section{The model}

A \emph{signed graph} $G$ is an undirected graph with signed edges, i.e.,
with edge weights equal to either $+1$ or $-1$.  Let
$\mcE=\mcE_+\union\mcE_-$ be the edge set of $G$, where $\mcE_+$ is the set
of positive edges and $\mcE_-$ the set of negative edges. $G$ is complete
when there exists an edge between any pair of vertices. A path
from vertex $i$ to $j$ in $G$ is a sequence of edges that connect a
sequence of distinct vertices starting from $i$ and finishing at $j$. 
A
connected component is any subgraph such that all of its vertices
are connected to each other by paths, but they are not connected to any
other vertex of $G$. $G$ is connected whenever it has only one connected
component.  The abbreviation $i.o.$ stands for \emph{infinitely
  often}.

We model the structure of a social network composed by agents as a graph. Then, throughout the paper, we use the words \emph{graph} and \emph{network} interchangeably, as well as the terms \emph{vertex} and \emph{agent}. 
Each agent in the network holds an opinion about a particular statement of a discussion topic, and her opinion describes how much she agrees with it. An agent $i$ has an opinion $x_i\in[o_{\min},o_{\max}]$: $x_i=o_{\max}$ 
whenever $i$ completely agrees with the statement being discussed, and $x_i=o_{\min}$ 
whenever she completely disagrees with it. The opinion vector $x\in[o_{\min},o_{\max}]^{n}$ 
has in its $i$th entry the opinion $x_i$ of agent $i$.


\begin{definition}[Sign arrangement property]
  Given a connected signed graph $G=(\until{n},\mcE_+\union\mcE_-)$ with $n\geq 3$, let
  $G_+=(\until{n},\mcE_+)$. For $k\in\N$, we say that $G$ satisfies the
  \emph{$k$-sign arrangement property} if
  \begin{enumerate}[label=(\roman*)]
  \item $G_+$ has $k\geq1$ connected components, and 
  \item each negative edge connects vertices belonging to different connected components of $G_+$.
  \end{enumerate}
  If this property holds, then each connected component of $G_+$ is a
  \emph{faction}.
\end{definition}

Based on the works~\cite{DC-FH:56,JD:67} in the sociological literature, we definite the notion of \emph{structural} and \emph{clustering balance} for connected graphs.

%
%
%
%
\begin{definition}[Structural and clustering balance]
  Consider a connected signed graph $G$ with $n\geq 3$.
  Assume the vertices of $G$ can be partitioned in $m$ groups such that
  each positive edge joins two vertices from the same group and each
  negative edge joins vertices from different groups. We say that $G$
  satisfies 
  \begin{enumerate}[label=(\roman*)]
  \item \emph{structural balance} if $m\leq 2$,  and 
  \item \emph{clustering balance} if $m\geq 3$.
  \end{enumerate}
\end{definition}

The following result follows immediately from the previous definitions.

\begin{lemma}
Let $G$ be a complete signed graph. $G$ satisfies the $k$-sign arrangement if and only if it satisfies structural balance when $k\leq 2$ or clustering balance when $k\geq 3$.   
\end{lemma}

Note that a signed graph satisfying the $k$-sign arrangement property does
not need to be complete.

\begin{definition}[Affine boomerang model]
  Let $G=(\until{n},\mcE_+\union\mcE_-)$ be a signed graph.
  Assume that each agent has an initial opinion $x_i(0)\in[o_{\min},o_{\max}]$, 
  $o_{\min}<o_{\max}$, and a self-weight $a_i\in(0,1)$.  At each time step
  $t\in\Z_{\geq0}$, select randomly an edge of $G$; assume each edge
  $\{i,j\}$ has a time-invariant positive selection probability $p_{ij}$.  Update
  the opinions of the two agents $i$ and $j$ according to:
  \begin{equation}\label{f1}
    x_i(t+1)=\begin{cases}
    a_i x_i(t)+(1-a_i) x_j(t), \qquad \text{ if } \{i,j\}\in\mcE_+,\\
    a_i x_i(t)+(1-a_i) o_{\min}, \\ \qquad\qquad \text{ if } \{i,j\}\in\mcE_- \text{ and } x_i(t)<x_j(t),\\
    a_i x_i(t)+(1-a_i) o_{\max},  \\ \qquad\qquad \text{ if } \{i,j\}\in\mcE_- \text{ and } x_i(t)\geq x_j(t),
    \end{cases}
  \end{equation}
  and similarly for agent $j$. 
\end{definition}

%

Note that the boomerang effect is captured by the last two cases
of equation~\eqref{f1}. 

We remark that our model has \emph{asynchronous updating} of the opinions since only two opinions are updated simultaneously and independently per time step instead of all opinions at once (which would be \emph{synchronous updating}). This type of updating has been present in other previous opinion models, e.g., in the Deffuant-Weisbuch model~\cite{GC-WS-WM-FB:18n}.
%
%
%
An example of selecting edges for the opinion updating is to do it uniformly as follows: let $m$ be the number of edges in the graph (e.g., $m={{n}\choose{2}}$ for complete graphs), then we can assign to every pair of agents the same probability of being selected and have 
$p_{ij}=1/m$ for any pair $\{i,j\}$.

\begin{remark}
  In our model, opinions take values on an arbitrary closed interval
  $[o_{\min},o_{\max}]$. From a sociological (and intuitive) point of view,
  it is plausible to have bounded opinions since there is no clear
  interpretation of a diverging opinion. Indeed, bounded opinions are
  present throughout the literature on opinion dynamics.  The case
  $o_{\min}=-\theta$ and $o_{\max}=\theta$, for $\theta>0$ is
  characteristic in the literature of bipartite consensus
  (e.g.,~\cite{CA:13,GS-CA-JB:19}), and the case $o_{\min}=0$ and
  $o_{\max}=1$ characterizes various works in the literature of opinion
  dynamics over graphs with positive weights (e.g.,~\cite{DA-AO:11}) or
  bounded-confidence models (e.g.,~\cite{GC-WS-WM-FB:18n}).
\end{remark}

\section{Model analysis}

\subsection{Theoretical results}
 \begin{theorem}[Consensus and polarization in signed graphs]  \label{th1}
  Consider a network satisfying the $k$-sign arrangement property. Consider
  the evolution of the affine boomerang model~\eqref{f1} with initial
  opinion vector $x(0)\in[o_{\min},o_{\max}]^n$.  Then
  \begin{enumerate}[label=(\roman*)]
  \item \label{l111} Consensus: if $k=1$, then, with
    probability one, $\lim_{t\to\infty}x(t)=c \vect{1}_n$,
    where $c$ is a random convex combination of the entries of
      $x(0)$.
  \item \label{l222} Polarization: if $k=2$, then, with
    probability one, $\lim_{t\to\infty}x_i(t)=o_{\min}$ for
    each agent $i$ of one of the two factions and
    $\lim_{t\to\infty}x_j(t)=o_{\max}$ for each $j$ of the
    other faction.
 \end{enumerate}
\end{theorem}
\begin{proof}
Formally, at any time step, the selected edge is a discrete random variable over some probability space $(\Omega',\mathcal{F}',{\prob}')$ with $\Omega'$ being the set of all edges on the graph, $\mathcal{F}'$ the power set, and 
${\prob}'[\{i,j\}]=p_{ij}$. Let $\omega(t)$ be the random edge selected at time $t$, then, the collection of random variables $\setdef{\omega(t)}{t\in\Z_{\geq 0}}$ forms a stochastic process of an independent sequence of random variables. Then, an adequate probability space $(\Omega,\mathcal{F},\prob)$ can be constructed with $\Omega=\prod_{t\in\Z_{\geq 0}}\Omega'$, $\mathcal{F}$ being the product of $\sigma$-algebras $\mathcal{F}'$ over $t\in\Z_{\geq 0}$, and $\prob$ being the product probability measure $\prod_{t\in\Z_{\geq 0}}{\prob}'$. Therefore, given the sequence of edges $\{s(t)\}_{t\in S}$ with some finite set $S\subset\Z_{\geq 0}$, $\prob[\setdef{\omega\in\Omega}{\omega(t)=s(t),t\in S}]=\prod_{t\in S}{\prob}'[s(t)]$. 

We start by considering the case $k=1$. In this case,
  the model is a linear system of the form $x(t+1)=W(t)x(t)$, where
  $W(t)$ is a random matrix that takes, at each time step, the value
  $W_{ij}=I_{n\times{n}}-(1-a_i)e_i(e_i-e_j)^\top-(1-a_j)e_j(e_j-e_i)^\top$
  whenever the edge $\{i,j\}$ is selected to be updated with
  probability $p_{ij}$ (here, $e_i$ is the $i$th column of the identity matrix $I_{n\times{n}}$). With probability one, $W(t)$ is a row
  stochastic matrix with a strictly positive diagonal for any $t$;
  moreover $W(t)$ is independent and identically distributed for any
  $t$. Thus, $\E[W(t)]$ (with respect to ${\prob}'$) 
  is a row stochastic matrix that, when
  interpreted as an adjacency matrix, corresponds to a connected undirected
  network.  Under these assumptions,~\cite[Theorem 12.1]{FB:18} implies
  the first statement of the theorem.


Now we prove the case $k=2$. 
We present 
the following lemma (whose proof is in the appendix) which describes the \emph{finite-time proximity property}:
\begin{lemma}[Finite-time proximity property]
\label{lemin1}
      Consider the same assumptions as in Theorem~\ref{th1}
      with a network satisfying the $2$-sign arrangement property.  There
      exists a finite sequence of edges such that, if they are
        selected sequentially by our affine boomerang model, then, inside
      the interval $[o_{\min},o_{\max}]$, the opinions of any two vertices
      become arbitrarily close if they belong to the same faction, or
      arbitrarily apart if they belong to different ones.
\end{lemma}
\noindent
%
%
%
%
Now, for any opinion vector $x\in[o_{\min},o_{\max}]^{n}$, we define the variable $Z:[o_{\min},o_{\max}]^n\to\{1,2\}$ as 
\begin{enumerate}[label=\textup{(C\arabic*)}]
\item $Z(x)=1$ when there is no value $\tau>0$ such that one faction has all of its opinions above $\tau$ and the other faction has them equal or below it; \label{22}
\item $Z(x)=2$ when there exists a value $\tau>0$ such that one faction has all of its opinions above $\tau$ and the other faction has them equal or below it. \label{23}
\end{enumerate}
Clearly, $Z$ exhausts all possible situations for the values of the opinion vector $x$, and, moreover, induces a partition over the set $[o_{\min},o_{\max}]^n$:  $[o_{\min},o_{\max}]^n=\cup_{m=1}^2Z^{-1}(m)$ and $Z^{-1}(1)\cap Z^{-1}(2)=\emptyset$.

Now, let us remark that, from the random selection process of the edges, it immediately follows that $\{x(t)\}_{t>0}$ is a random process over the probability space $(\Omega,\mathcal{F},\prob$); and, moreover, it is a Markov process, i.e., $\prob[x(t)\in Z^{-1}(M)\,|\,x(t-1)=c_{t-1},\dots,x(0)=c_o]=\prob[x(t)\in Z^{-1}(m)\,|\,x(t-1)=c_{t-1}]$ for any $m\in\{1,2\}$. Observe that, with probability one, $x(t)\in[o_{\min},o_{\max}]^{n}$ for any $t$ since $x(0)\in[o_{\min},o_{\max}]^{n}$.

Now, assume that $x(t)\in Z^{-1}(2)$ for some $t<\infty$. Let $F_1$ be the faction such that $x_i(t)\leq\tau$ for any $i\in F_1$; and $F_2$ the one such that $x_i(t)>\tau$ for any $i\in F_2$. Let $\theta_{F_1}(t)=\max_{i\in {F_1}}x_i(t)$ and $\theta_{F_2}(t)=\min_{i\in {F_2}}x_i(t)$. If at $t+1$ some $i\in F_1$ and $j\in F_2$ are selected, we have that $x_i(t+1)< x_i(t)$ and $x_j(t+1)>x_j(t)$. 
%
%
On the other hand, if at $t+1$ both $i$ and $j$ belong to the same faction with $x_i(t)\leq x_j(t)$, we have that $x_i(t)\leq x_i(t+1),x_j(t+1)\leq x_j(t)$, with equality if and only if $x_i(t)=x_j(t)$. From these two observations it is easy to show that $\theta_{F_1}(t+1)\leq\theta_{F_1}(t)$ with probability one; i.e., $\{\theta_{F_1}(s)\}_{s\geq t}$ is a non-decreasing sequence which is lower bounded by $o_{\min}$. This implies convergence of $\{\theta_{F_1}(s)\}_{s\geq t}$ to some lower bound $c_{\min}$ with probability one. 
Now, for any $\epsilon>0$ and $t^*\geq t$, there exists some finite $T>0$ such that if the sequence of edges $\{(\theta_{F_1}(s),k(s))\}_{s=t^*}^{t^*+T}$ with $k(s)\in F_2$ for $t^*\leq s\leq t^*+T$ is selected, then $|\theta_{F_1}(t^*+T)-o_{\min}|<\epsilon$. Such sequence has a positive probability of being selected sequentially by the affine boomerang model 
for any $t^*$, from which it follows that $c_{\min}=o_{\min}$.
%
%
%
Therefore, 
there is polarization for any $i\in F_1$ towards $o_{\min}$. A similar reasoning leads to the proof that $\{\theta_{F_2}(s)\}_{s\geq t}$ has an analogous increasing monotonic behavior and thus that there is polarization for $i\in F_2$ towards $o_{\max}$ with probability one. 
In conclusion, if $x(t)\in Z^{-1}(2)$ for $t\geq 0$, then polarization occurs with probability one and we say that $Z^{-1}(2)$ is an \emph{absorbing set} since the opinion vector cannot escape from it once it enters this set.

Therefore, to finish the proof of the theorem, we only need to prove that, given $x(t)\in Z^{-1}(1)$ at any time $t$, there always exists (with probability one) a finite sequence of edges 
such that eventually $x(t^*)\in Z^{-1}(2)$ for some $t<t^*<\infty$. Then, since any finite 
sequence of edges has positive probability of being selected sequentially by the affine boomerang model and $Z^{-1}(3)$ is an absorbing set, 
it follows 
that $\prob[x(t)\in Z^{-1}(1)\; i.o.\,|\,x(0)\in Z^{-1}(1)]=0$;
and this finishes the proof for item~\ref{l222} of the theorem. 
Therefore, it suffices to prove that
$\prob[x(t+T)\in Z^{-1}(2)\text{ for some } T>0\,|\,x(t)=x_o]=1$ for any $x_o\in Z^{-1}(1)$.
%
So, let us fix any $x_o\in Z^{-1}(1)$. Let $\mathcal{T}_{1\to 2}(t)=\inf\{t^*>t:x(t^*)\in Z^{-1}(2)|x(t)=x_o\}$ be the first time, after starting in $x_o\in Z^{-1}(1)$ at time $t$, at which the opinion vector enters the set $Z^{-1}(2)$. If we show that $\prob[\mathcal{T}_{1\to 2}(t)<\infty]=1$ for any $t$, then we have finished the proof. 

By the Markov property, 
%
we only need to show that $\prob[\mathcal{T}_{1\to 2}(0)<\infty]=1$. We start by noticing that, by the finite-time proximity property, there exists a sequence of edges $s(0),\dots,s(\tau-1)$ for some $\tau>0$ such that $x(\tau)\in Z^{-1}(2)$. Let $\gamma_o:=\min_{\{i,j\}\in\mcE}p_{ij}$. Then,
\begin{equation}
\label{unin}
\begin{split}
\prob[&x(\tau)\in Z^{-1}(2)|x(0)=x_o]\\
&\geq\prob[\omega(0)=s(0)|x(0)=x_o]\\
&\phantom{= }\;\times\prob[\omega(1)=s(1) \,|\, x(0)=x_o,\omega(0)=s(0)]\dots\\
&\phantom{= }\;\times\prob[\omega(\tau-1)=s(\tau-1)\}\,|\\
&\qquad\qquad x(0)=x_o,\omega(\ell)=s(\ell)\text{ for } \ell\in[0,\tau-2]]\\
&={\prob}'[s(0)]{\prob}'[s(1)]\dots{\prob}'[s(\tau-1)]\\
&\geq(\gamma_o)^{\tau},
\end{split}
\end{equation}
where the first inequality comes from a repetitive application of the conditional probability and the following equality comes from the independence of the underlying stochastic process. 
Let $\Gamma>0$ be any integer and $A_{\ell}=\{x(t)\notin Z^{-1}(2),t\in[\ell,\ell+\tau]\}$, then $\prob[A_{0}|x(0)=x_o]\leq 1-\gamma_o^\tau$. Likewise, in a  way similar to how we obtained expression~\eqref{unin}, we compute
\begin{equation*}
\begin{split}
\prob[&\mathcal{T}_{1\to 2}(0)\geq(\tau+1)\Gamma]\\
&=\prob[x(t)\notin Z^{-1}(2),t\in[0,(\tau+1)\Gamma-1]|x(0)=x_o]\\
&=\prob[\cap_{\ell=0}^{\Gamma-1}A_{\ell(\tau+1)}|x(0)=x_o]\\
&=\prob[A_{0}|x(0)=x_o]\\
&\phantom{=} \quad \times \prod_{\ell=1}^{\Gamma-1}\prob[A_{\ell(\tau+1)}|x(0)=x_o,\cap_{0\leq \ell'\leq \ell}A_{\ell'(\tau+1)}]\\
&\leq(1-\gamma_o^{\tau})^\Gamma=:\gamma^\Gamma.
\end{split}
\end{equation*}
Now, we observe that $\sum_{t=1}^{\infty}\prob[\mathcal{T}_{1\to 2}(0)\geq(\tau+1)t]\leq\sum_{t=1}^{\infty}\gamma^t=\frac{\gamma}{1-\gamma}<\infty$ because of geometric series since $0<\gamma<1$. Then, by the first Borel-Cantelli lemma, we conclude that $\prob[\mathcal{T}_{1\to 2}(0)<\infty]=1$. This concludes the proof.
%
%
\end{proof}

A consequence of Theorem~\ref{th1} is that a complete social network that satisfies
structural balance with 
two factions ends up having its agents with totally
opposite opinions. This agrees with the intuitive result that antagonistic
groups are expected to develop polarized opinions, as shown by other models in the
literature~\cite{YL-WC-YW-ZZ:15,JL-XC-TB-MAB:17}. Also, as expected, if
there are no negative relationships between the agents (i.e., there is only
one faction), all agents reach consensus.
\begin{lemma}[Fluctuations]  \label{th2}
  Consider a network satisfying the $k$-sign arrangement property with
  $k\geq3$ factions $\{F_1,\dots,F_k\}$ and such that there
  exists at least one negative edge between any pair of
  factions.  Consider the boomerang opinion dynamics model~\eqref{f1}
  with $x_i(0)=o_{\min}$ for any $i\in F_1$,
  $x_i(0)=o_{\max}$ for any $i\in F_2$, and
  $x_i(0)\in(o_{\min},o_{\max})$ for any $i\in F_k$, $k\geq3$.  Then,
  for any $0<\epsilon<({o_{\max}-o_{\min}})/{2}$ and any
    $i\in F_k$, $k\geq3$,
  \begin{equation*}
    \prob[x_i(t)\in(o_{\min},o_{\min}+\epsilon)\cup(o_{\max}-\epsilon,o_{\max})\textup{ i.o.}]=1.
  \end{equation*} 
\end{lemma}
  \begin{proof}
    Note that $x_i(t)\in(o_{\min},o_{\max})$ for any $t\geq 0$ and any
    $i\in F_k$, $k\geq 3$, with probability one.  Pick a positive
    $\epsilon<({o_{\max}-o_{\min}})/{2}$ and define the intervals
    $A_{\epsilon}^\ell=(o_{\min},o_{\min}+\epsilon)$,
    $A_{\epsilon}^u=(o_{\max}-\epsilon,o_{\max})$ and
    $A_{\epsilon}^c=[o_{\min}+\epsilon,o_{\max}-\epsilon]$. Note that these
    three intervals are non-empty and form a partition of
    $(o_{\min},o_{\max})$.

    Now, take any $i\in F_k$, $k\geq 3$. First, define the random stopping
    times $\tau_{c\to \ell}(t)=\inf\{t^*>t\,|\,x_i(t^*)\in
    A_{\epsilon}^\ell|x_i(t)\in A_{\epsilon}^c\}$, $\tau_{\ell\to
      u}(t)=\inf\{t^*>t\,|\,x_i(t^*)\in A_{\epsilon}^u|x_i(t)\in
    A_{\epsilon}^\ell\}$ and $\tau_{u\to \ell}(t)=\inf\{t^*>t\,|\,x(t^*)\in
    A_{\epsilon}^\ell|x_i(t)\in A_{\epsilon}^u\}$.  Note that, if the pair
    $\{i,j\}$ is chosen, then the opinion of $i$ is always pushed towards
    $o_{\max}$ if $j\in F_1$, and always pushed towards $o_{\min}$ if $j\in
    F_2$ (this follows from the fact that for any $k\in F_1\cup F_2$,
    $x_k(t)=x_k(0)$ for all $t\geq 0$ with probability one). Then,
    following a reasoning similar to the one adopted in the proof of
    Theorem~\ref{th1}, we conclude that $\prob[\tau_{c\to\ell}(t)<\infty] =
    \prob[\tau_{\ell\to u}(t)<\infty]=\prob[\tau_{u\to\ell}(0)<\infty]=1$
    for any $t\geq 0$, from which the result follows.
\end{proof}

Note that the conditions for the underlying signed network in this
lemma are immediately satisfied if the network is complete
and satisfies clustering balance. This lemma is interpreted
as follows. Assume there are multiple antagonistic groups
of people such that for any two groups there exist two members
  that can communicate with each other. Additionally, assume
  that only two groups are already polarized in the opinion spectrum
with the rest having opinions at intermediate values (i.e.,
  mathematically, in the interval $(o_{\min},o_{\max})$).
Then, these non-polarized groups will have their opinions
  always fluctuating at intermediate values, i.e., their opinions do not
  polarize or reach consensus at some specific value.  Intuitively, since
the boomerang effect is persistent on the agents with
  intermediate values, these agents cannot settle on a definite opinion
since they continue to interact with antagonistic agents on both ends of
the spectrum. This behavior of opinion fluctuation has been
observed in other models in the presence of stubborn agents who forbid the
consensus of opinions among the agents~\cite{DA-GC-FF-AO:10}. Our work is
the first one to propose a persistent fluctuating behavior based on the
structure of friendly and antagonistic relationships in a social network.

\subsection{Numerical results}

For a complete graph satisfying structural balance, which is a particular case
satisfying the conditions of Theorem~\ref{th1}, Figure~\ref{f:sim1} shows
some example evolutions for self-weights $a_i~=~a~\in~(0,1)$ for any agent
$i$. We observed that, in general, the larger the self-weights, the more
time the polarization process takes.

Figure~\ref{f:s2} shows examples where the underlying signed network has three factions. Remarkably, under generic initial conditions (which are weaker initial conditions than the ones in Corollary~\ref{th2}), two factions tend to polarize and the opinions of the third one show persistent fluctuations.
%
%
%
%

Finally, we provide numerical evidence of the behavior under networks that
are the result of perturbations on balanced networks.  Consider the
situation where a complete and balanced social network with two
antagonistic factions is randomly perturbed by flipping the sign of some of
its edges. Intuitively, for small perturbations, we would expect that
opinions, though not being able to perfectly polarize, would still
``attempt'' to be in such a state and fluctuate near extreme
values. Figure~\ref{f:sim3} shows some examples confirming this phenomenon.

%
%
\begin{figure}[t]
  \centering
  \subfloat[$a=0.25$]{\includegraphics[width=0.33\linewidth]{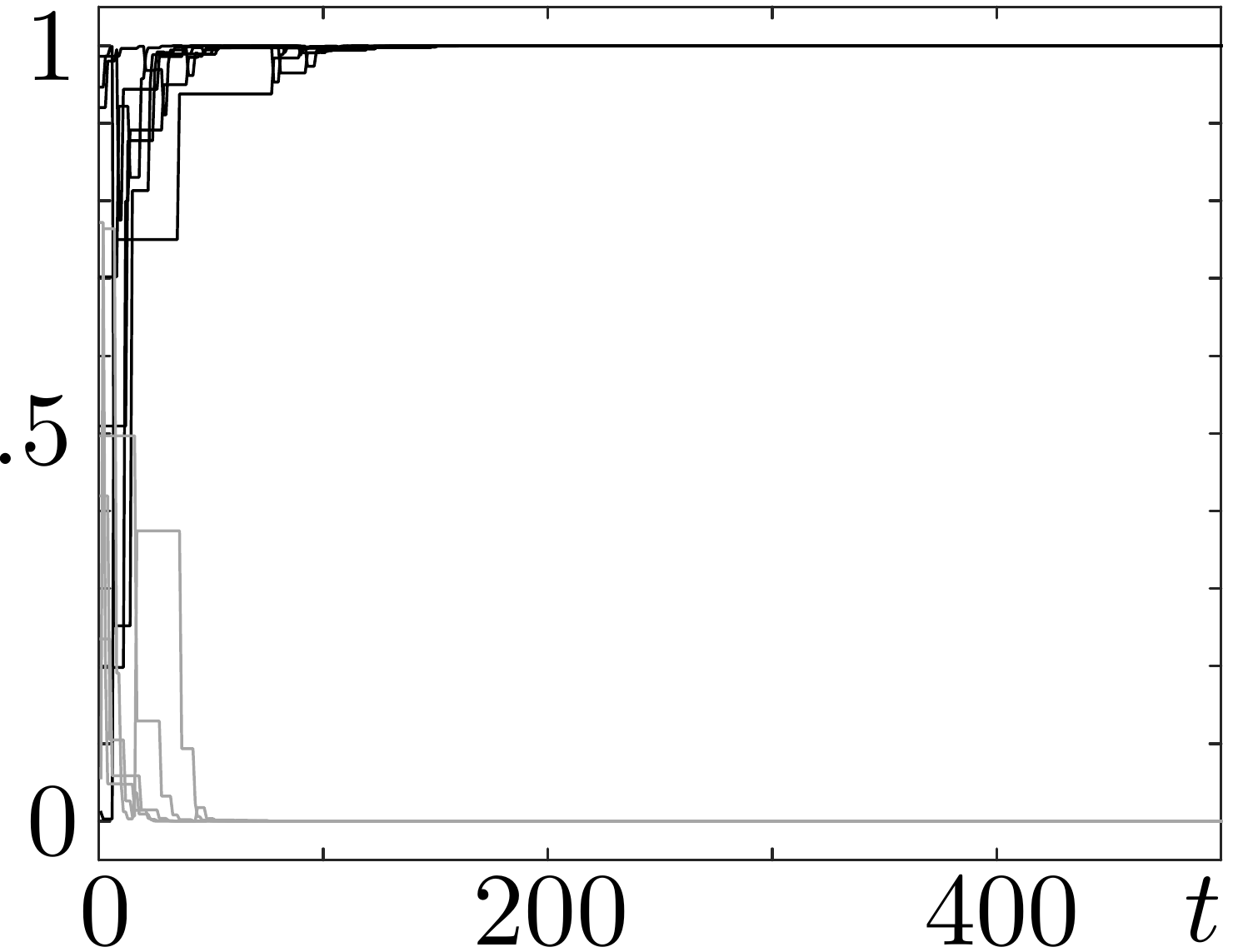}} 
  \subfloat[$a=0.5$]{\includegraphics[width=0.33\linewidth]{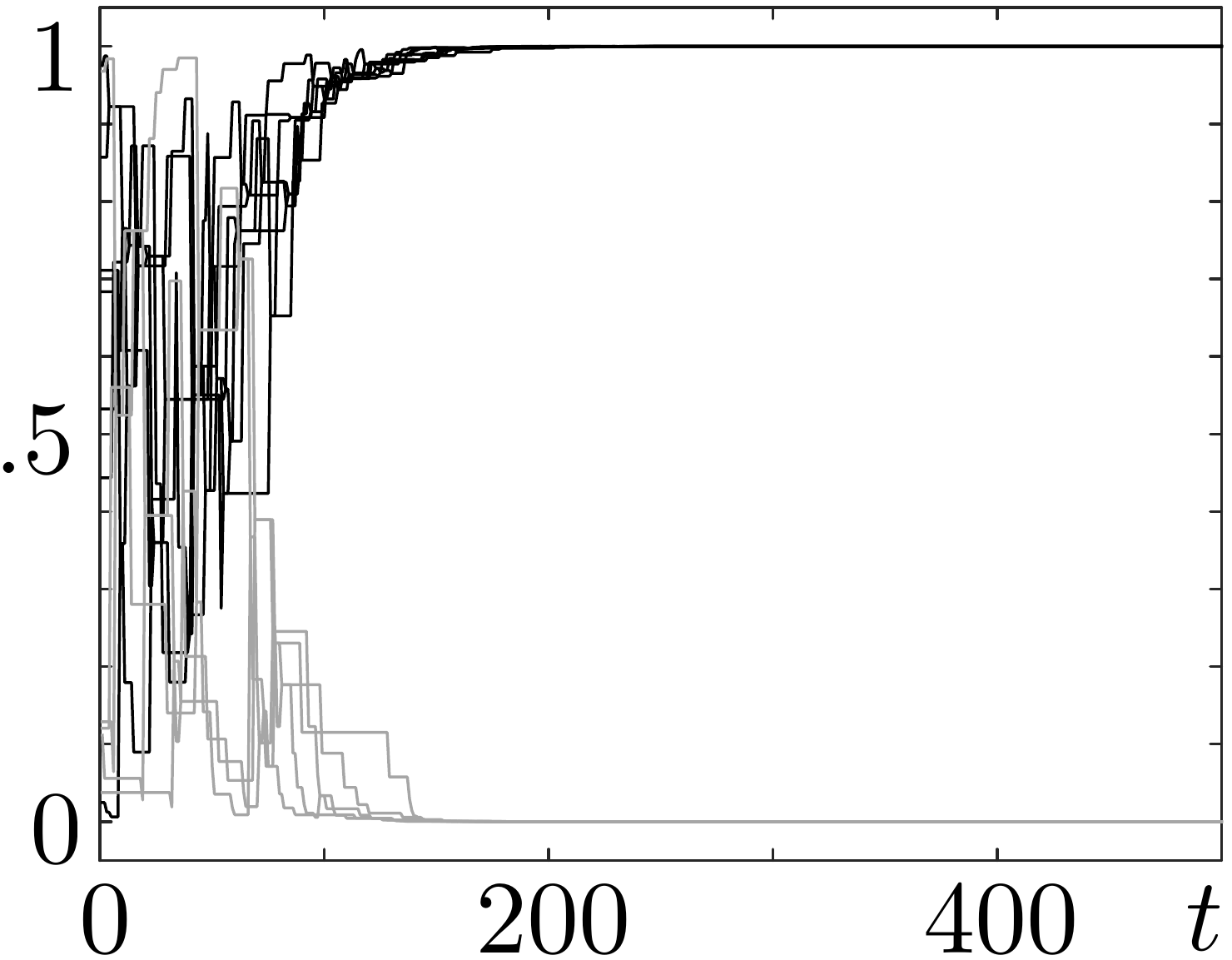}}
  \subfloat[$a=0.75$]{\includegraphics[width=0.33\linewidth]{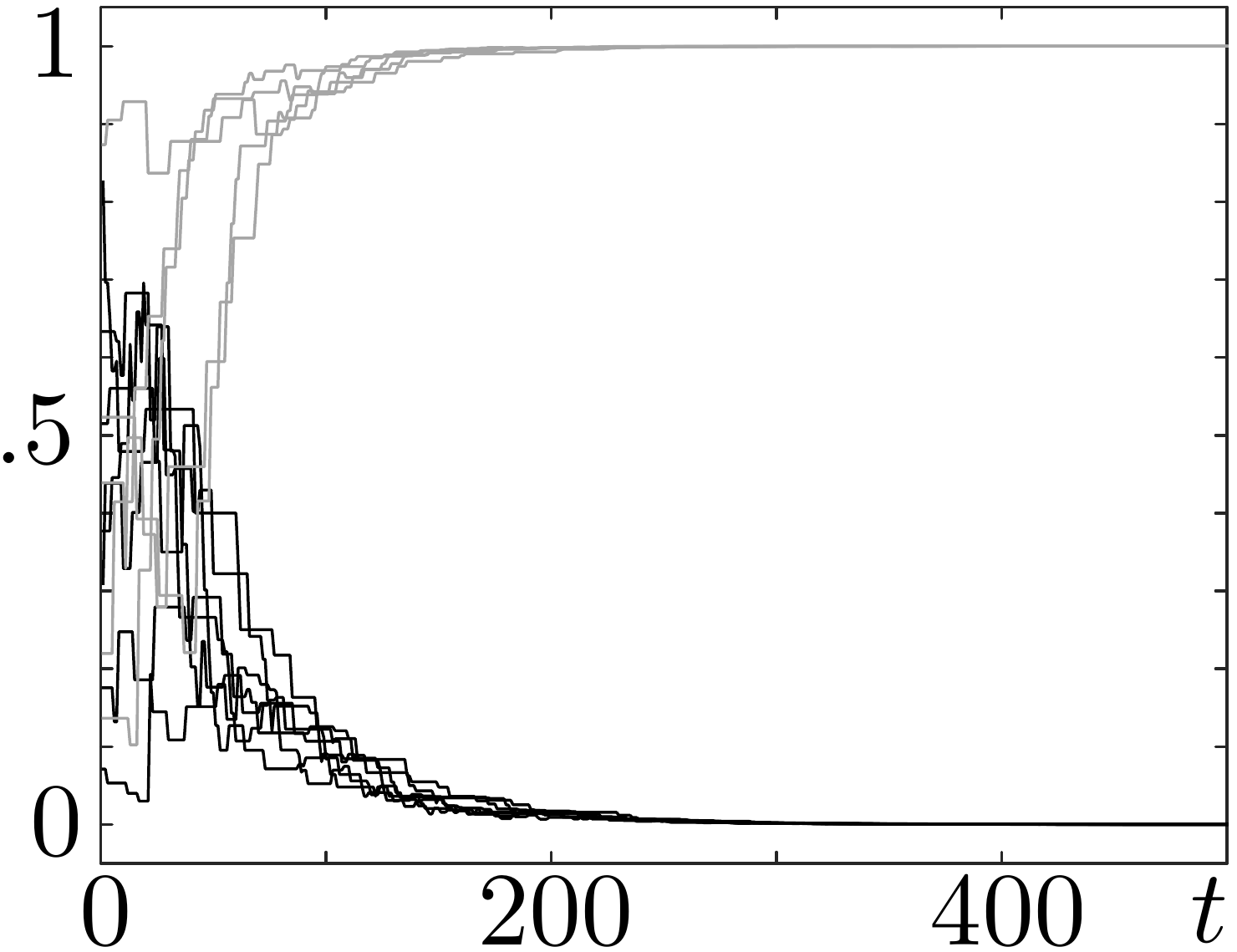}} 
\caption{Opinion evolution with $o_{\min}=0$ and $o_{\max}=1$ for a complete graph satisfying structural balance 
with two factions of $5$ (light gray) and $7$ (black) agents respectively. All agents are assumed to have the same self-weight $a$, and the edges to be updated are chosen uniformly. All simulations have randomly sampled initial conditions.}
  \label{f:sim1}
\end{figure}

%
 \begin{figure}[t]
   \centering
   \subfloat[$a=0.25$]{\includegraphics[width=0.33\linewidth]{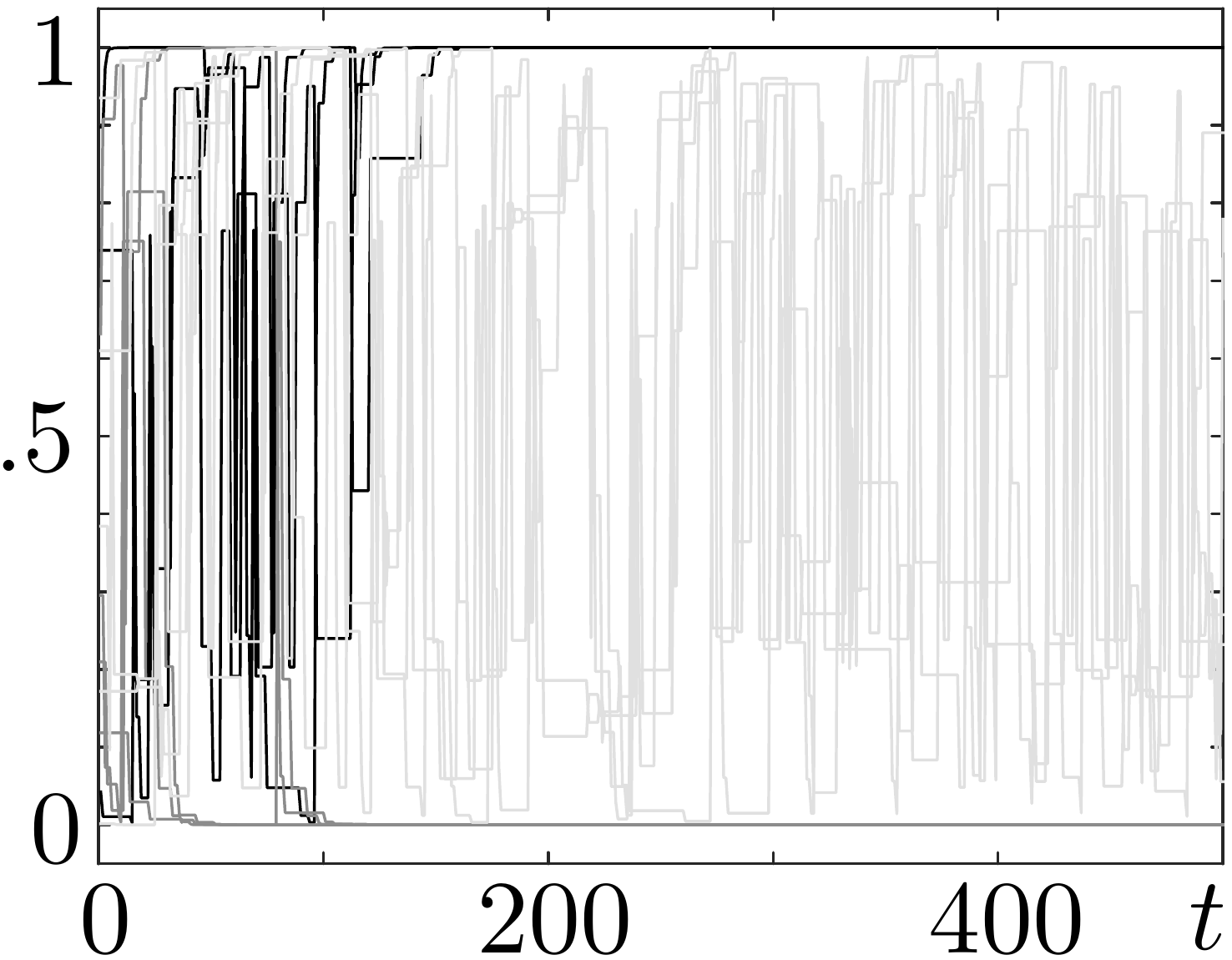}} 
  \subfloat[$a=0.5$]{\includegraphics[width=0.33\linewidth]{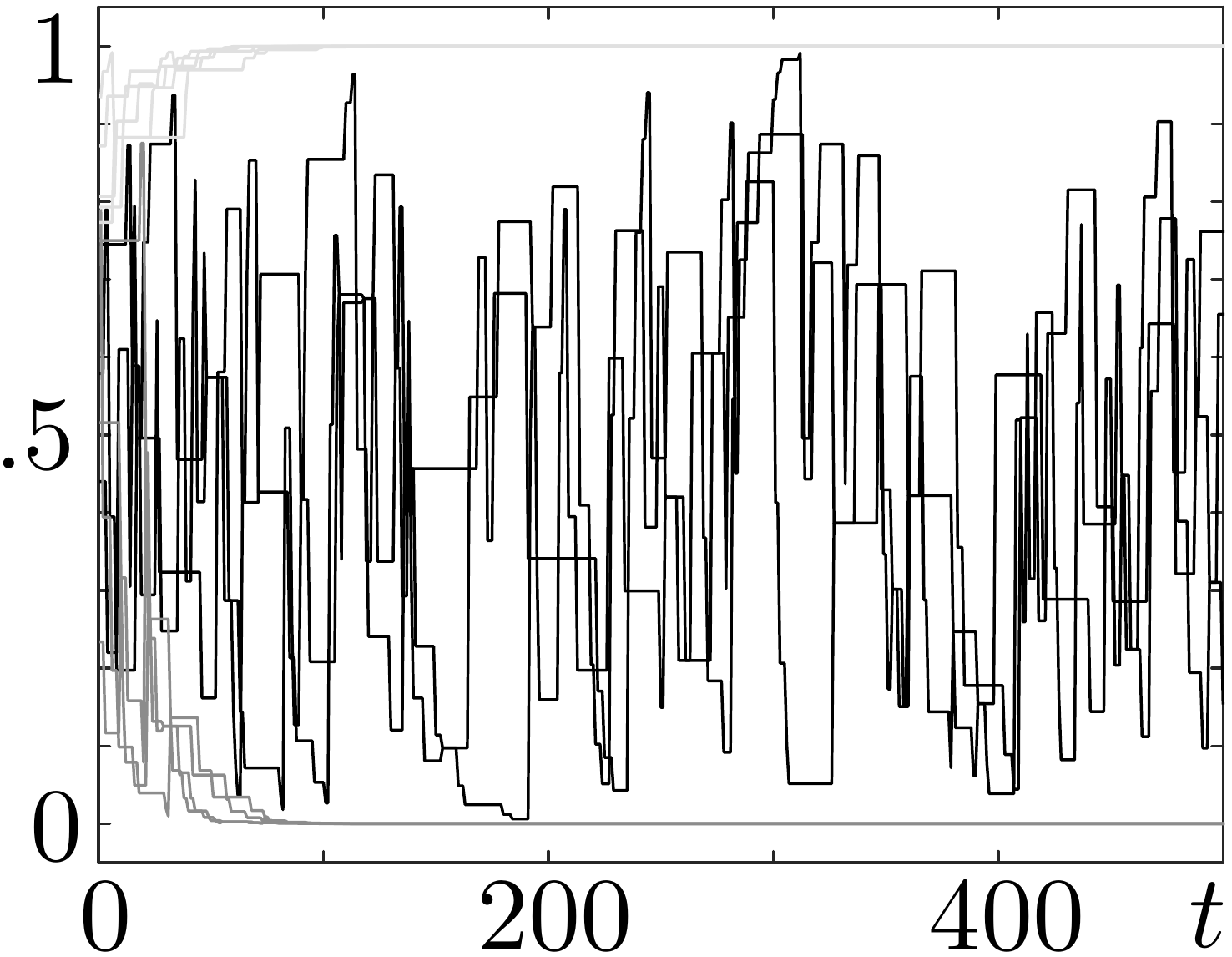}}
  \subfloat[$a=0.75$]{\includegraphics[width=0.33\linewidth]{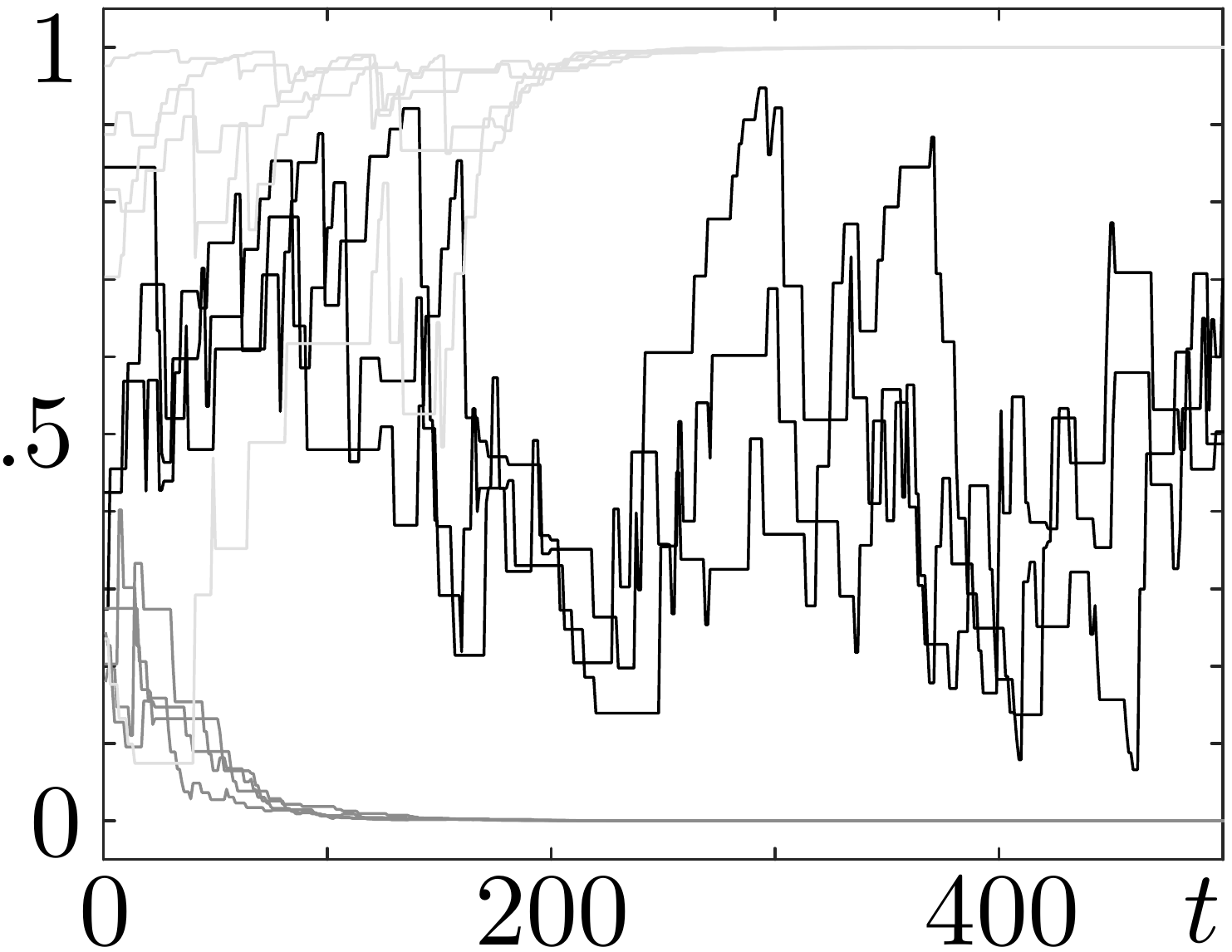}}
 \caption{Opinion evolution with $o_{\min}=0$ and $o_{\max}=1$ for a complete graph satisfying clustering balance 
 with three clusters of three, four and five agents (i.e., twelve curves are plotted). The black curves correspond to the opinions of the cluster of three agents, the medium gray curves to the cluster of four, and the light gray curves to the cluster of five. Two of the clusters polarized their opinions (to $0$ and $1$), while the third one shows permanent fluctuations in its opinions. The shown plots were chosen so that the cluster with four agents always end up oscillating. All agents are assumed to have the same self-weight $a$, and the edges to be updated are chosen uniformly. All simulations have randomly sampled initial conditions.}
    \label{f:s2}
 \end{figure}

\begin{figure}[t]
  \centering
  \subfloat[$a=0.25$]{\includegraphics[width=0.33\linewidth]{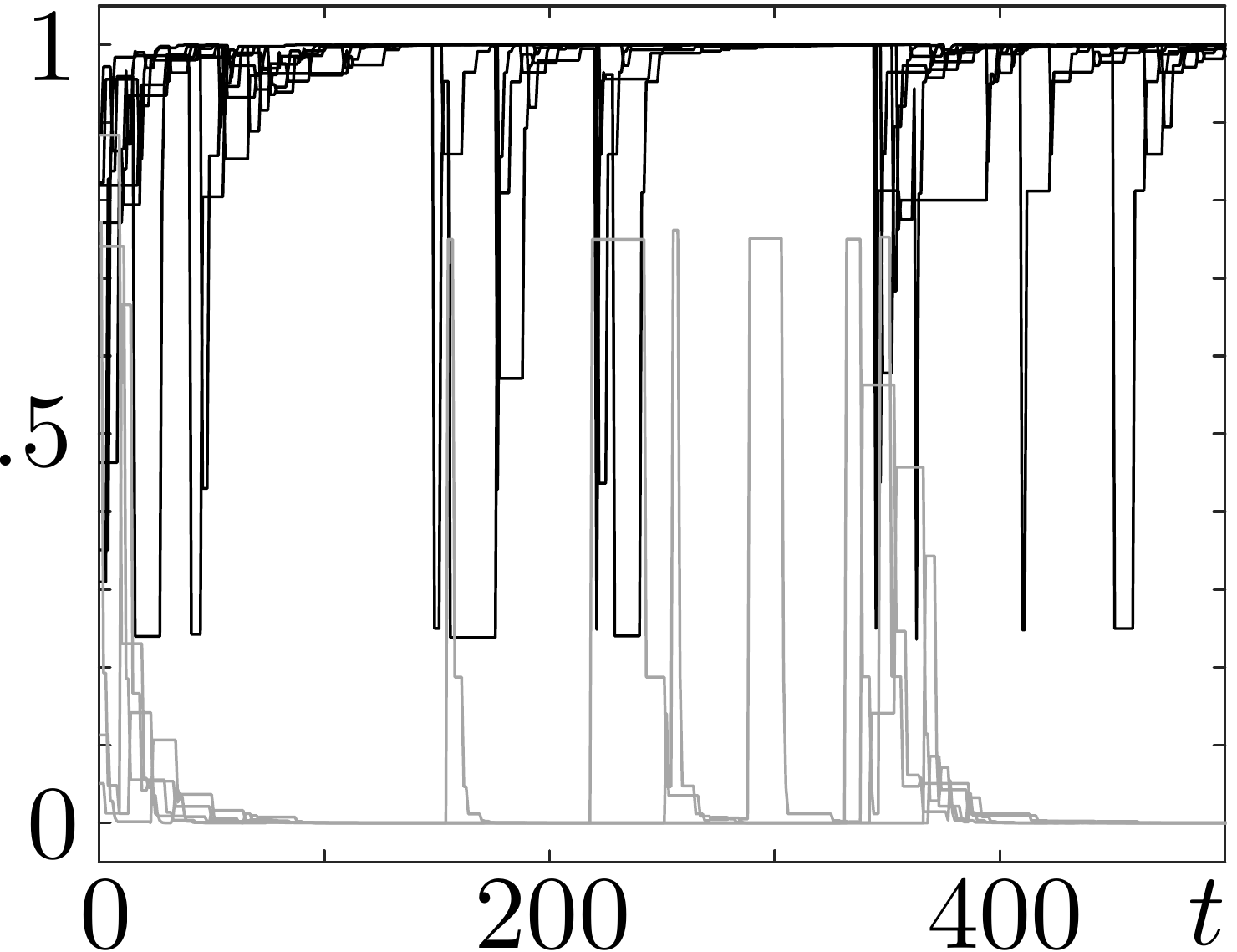}} 
  \subfloat[$a=0.5$]{\includegraphics[width=0.33\linewidth]{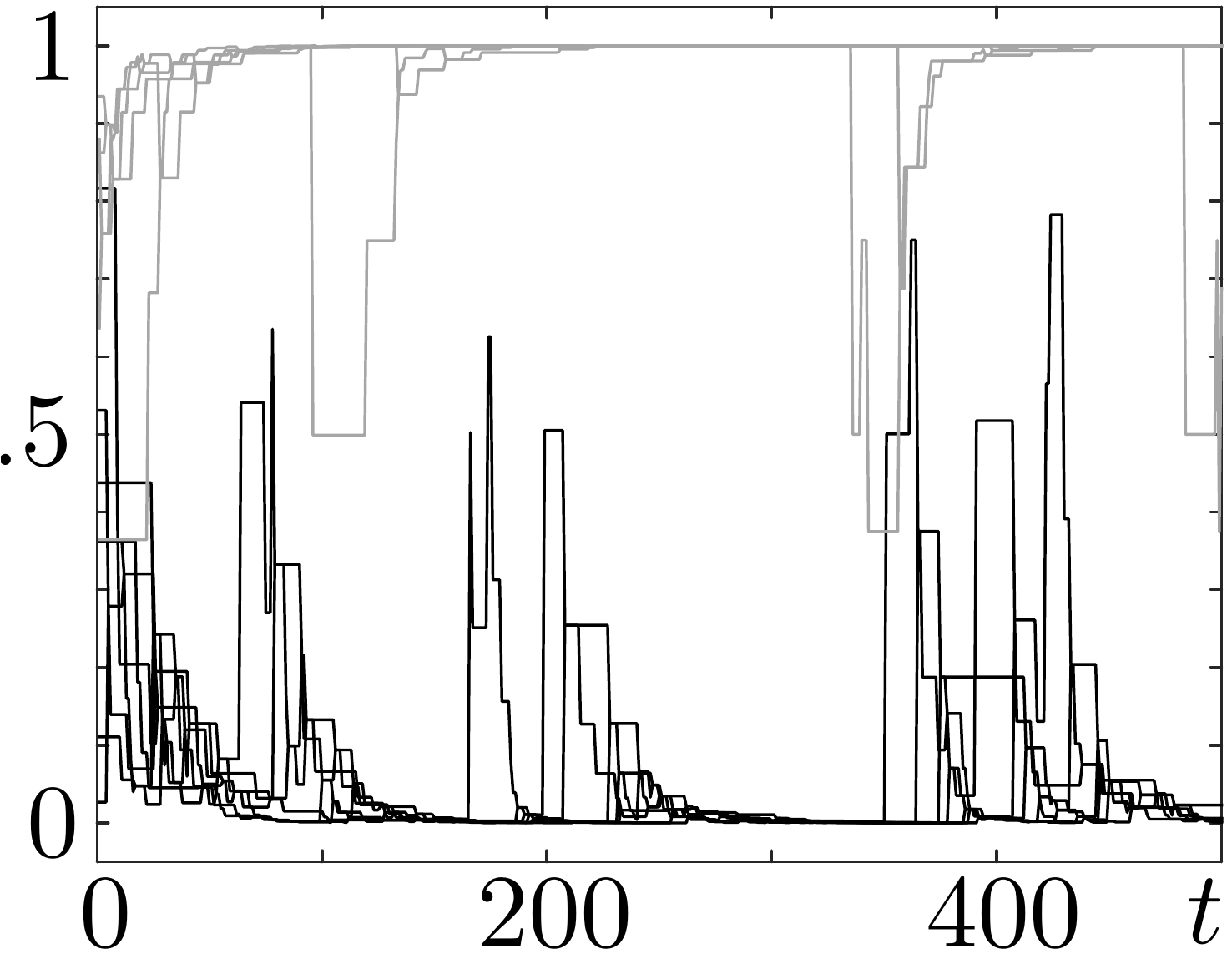}}
  \subfloat[$a=0.75$]{\includegraphics[width=0.33\linewidth]{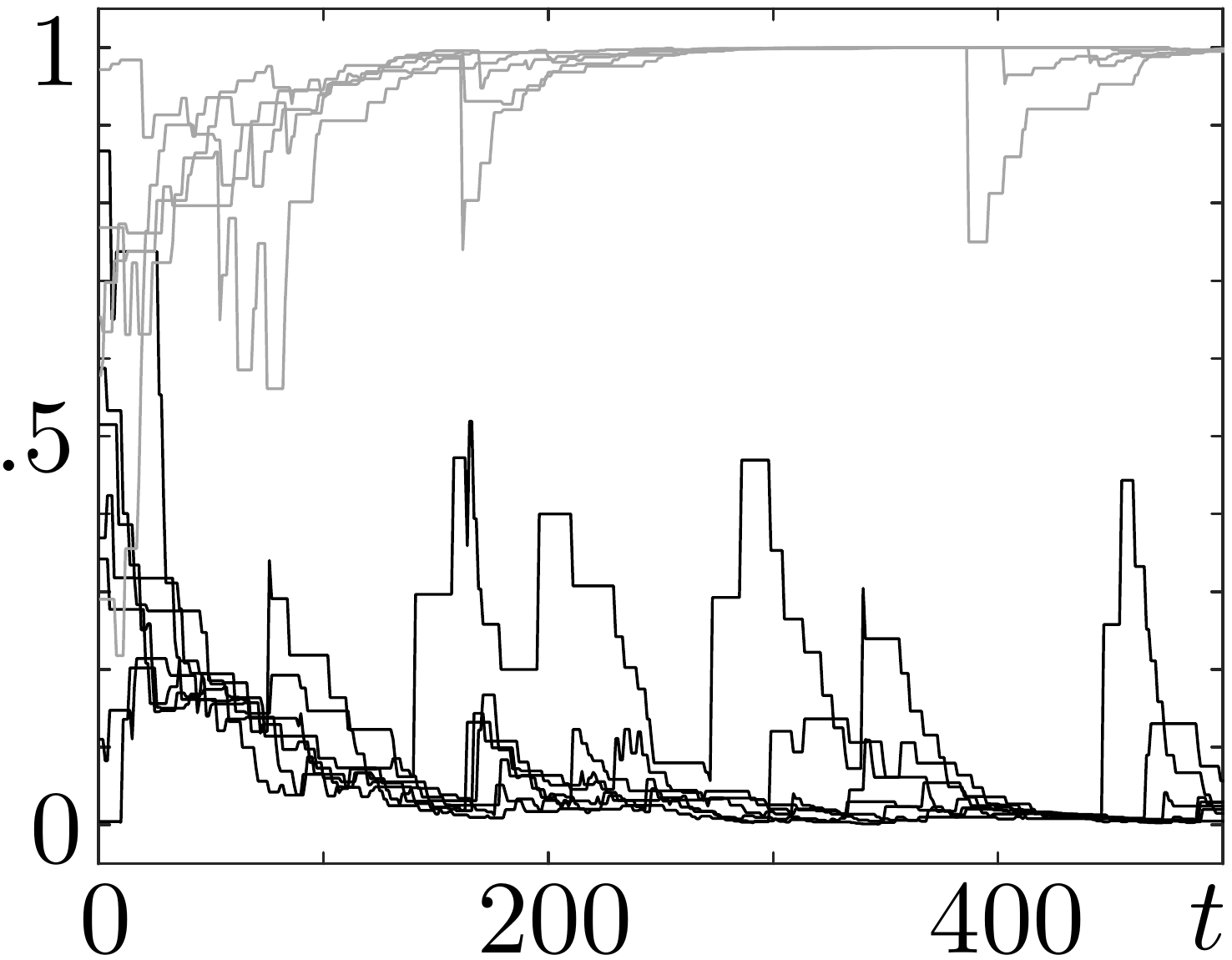}} 
  \caption{Opinion evolution with $o_{\min}=0$ and $o_{\max}=1$ for a complete graph 
  that originally satisfied structural balance with 
  two factions of $5$ (light gray) and $7$ (black) agents and is now under a perturbation of $3$ of its edges having the opposite sign. All agents are assumed to have the same self-weight $a$, and the edges to be updated are chosen uniformly. All simulations have randomly sampled initial conditions.}
  \label{f:sim3}
\end{figure}

\section{Conclusion}
We have proposed a novel simple model for opinion dynamics over signed graphs. This model provides intuitive behavior and results on the opinion evolution under sociologically relevant sign structures of the underlying social network. Future work may be the inclusion of directional updating (i.e., updating one opinion at a time) in the model, as well as its analysis under relevant directed network structures. Another open direction for research is an analytical understanding of the transient time and convergence analysis for the polarization of opinions of the factions in a balanced network.

\appendix

\begin{proof}[Proof for Lemma~\ref{lemin1}]
%

Since the network satisfies the 2-sign arrangement, for any $i$ and $j$ that belong to the same faction, there exists a nonempty collection of 
paths $\mathcal{P}^+_{i\leftrightarrow	j}$ between $i$ and $j$ in which each path contains only positive edges. Let $p\in\mathcal{P}^+_{i\leftrightarrow j}$, then, from statement~\ref{l111} from Theorem~\ref{th1}, we observe that if we only update pair of vertices present along the path $p$, then they can become arbitrarily close. Then, we can construct a finite sequence of edges such that it includes only edges from one or more different paths in $\mathcal{P}^+_{i\leftrightarrow j}$ in a sufficient number so that $i$ and $j$ become arbitrarily close. 
This proves the first part of the lemma. 

Now, we consider the case where $i$ and $j$ belong to different factions. Notice that equation~\eqref{f1} clearly shows that we can always make the opinions of two vertices joined by a negative edge arbitrarily apart by continuously sampling such edge. 
Let $\mathcal{P}^-_{i\leftrightarrow j}$ be the nonempty collection of 
paths between $i$ and $j$. Due to the structure of the network, any $p\in\mathcal{P}^-_{i\leftrightarrow j}$ must have an odd number of negative edges. Then, $p$ can be constructed by appropriately concatenating sequences of positive edges with sequences of negative edges. 
%
%
%
%
%
%
%
 %
%
%
%
From our discussion above, we can make the opinions of the agents participating in any of these positive sequences (if any) arbitrarily close, and the opinions of the agents in any of the negative edges arbitrarily apart. Then, it is possible to come up with a finite sequence of edges such that $i$ and $j$ become arbitrarily apart. This finishes the proof of the lemma. 
\end{proof}

\bibliographystyle{plainurl}
\bibliography{alias,Main,FB}

\begin{thebibliography}{10}

\bibitem{RPA-JCM:67}
R.~P. Abelson and J.~C. Miller.
\newblock Negative persuasion via personal insult.
\newblock {\em Journal of Experimental Social Psychology}, 3(4):321--333, 1967.
\newblock \href {http://dx.doi.org/10.1016/0022-1031(67)90001-7}
  {\path{doi:10.1016/0022-1031(67)90001-7}}.

\bibitem{DA-GC-FF-AO:10}
D.~Acemoglu, G.~Como, F.~Fagnani, and A.~Ozdaglar.
\newblock Opinion fluctuations and disagreement in social networks.
\newblock {\em Mathematics of Operation Research}, 38(1):1--27, 2013.
\newblock \href {http://dx.doi.org/10.1287/moor.1120.0570}
  {\path{doi:10.1287/moor.1120.0570}}.

\bibitem{DA-AO:11}
D.~Acemoglu and A.~Ozdaglar.
\newblock Opinion dynamics and learning in social networks.
\newblock {\em Dynamic Games and Applications}, 1(1):3--49, 2011.
\newblock \href {http://dx.doi.org/10.1007/s13235-010-0004-1}
  {\path{doi:10.1007/s13235-010-0004-1}}.

\bibitem{CA:13}
C.~Altafini.
\newblock Consensus problems on networks with antagonistic interactions.
\newblock {\em IEEE Transactions on Automatic Control}, 58(4):935--946, 2013.
\newblock \href {http://dx.doi.org/10.1109/TAC.2012.2224251}
  {\path{doi:10.1109/TAC.2012.2224251}}.

\bibitem{NJS-PBB:01}
A.~Aron and G.~Lewandowski.
\newblock Psychology of interpersonal attraction.
\newblock In N.~J. Smelser and P.~B. Baltes, editors, {\em International
  Encyclopedia of the Social \& Behavioral Sciences}, pages 7860--7862.
  Pergamon, 2001.
\newblock \href {http://dx.doi.org/10.1016/B0-08-043076-7/01787-3}
  {\path{doi:10.1016/B0-08-043076-7/01787-3}}.

\bibitem{FB:18}
F.~Bullo.
\newblock {\em Lectures on Network Systems}.
\newblock Kindle Direct Publishing, 1.3 edition, July 2019.
\newblock With contributions by J. Cort{\'e}s, F. D\"orfler, and S.
  Mart{\'\i}nez.
\newblock URL: \url{http://motion.me.ucsb.edu/book-lns}.

\bibitem{SB-PSH:09}
S.~Byrne and P.~{Solomon~Hart}.
\newblock The boomerang effect. {A} synthesis of findings and a preliminary
  theoretical framework.
\newblock {\em Annals of the International Communication Association},
  33(1):3--37, 2009.
\newblock \href {http://dx.doi.org/10.1080/23808985.2009.11679083}
  {\path{doi:10.1080/23808985.2009.11679083}}.

\bibitem{DC-FH:56}
D.~Cartwright and F.~Harary.
\newblock Structural balance: {A} generalization of {H}eider's theory.
\newblock {\em Psychological Review}, 63(5):277, 1956.
\newblock \href {http://dx.doi.org/10.1037/h0046049}
  {\path{doi:10.1037/h0046049}}.

\bibitem{GC-WS-WM-FB:18n}
G.~Chen, W.~Su, W.~Mei, and F.~Bullo.
\newblock Convergence properties of the heterogeneous {Deffuant-Weisbuch}
  model.
\newblock {\em Automatica}, January 2019.
\newblock Submitted.

\bibitem{ARC:62}
A.~R. Cohen.
\newblock A dissonance analysis of the boomerang effect.
\newblock {\em Journal of Personality}, 30(1):75--88, 1962.
\newblock \href {http://dx.doi.org/10.1111/j.1467-6494.1962.tb02306.x}
  {\path{doi:10.1111/j.1467-6494.1962.tb02306.x}}.

\bibitem{JD:67}
J.~A. Davis.
\newblock Clustering and structural balance in graphs.
\newblock {\em Human Relations}, 20(2):181--187, 1967.
\newblock \href {http://dx.doi.org/10.1177/001872676702000206}
  {\path{doi:10.1177/001872676702000206}}.

\bibitem{JMH:14}
J.~M. Hendrickx.
\newblock A lifting approach to models of opinion dynamics with antagonisms.
\newblock In {\em {IEEE} Conf.\ on Decision and Control}, pages 2118--2123,
  December 2014.
\newblock \href {http://dx.doi.org/10.1109/CDC.2014.7039711}
  {\path{doi:10.1109/CDC.2014.7039711}}.

\bibitem{CIH-OJH-MS:57}
C.~I. Hovland, O.~J. Harvey, and M.~Sherif.
\newblock Assimilation and contrast effects in reactions to communication and
  attitude change.
\newblock {\em The Journal of Abnormal and Social Psychology}, 55(2):244--252,
  1957.
\newblock \href {http://dx.doi.org/10.1037/h0048480}
  {\path{doi:10.1037/h0048480}}.

\bibitem{YL-WC-YW-ZZ:15}
Y.~Li, W.~Chen, Y.~Wang, and Z.-L. Zhang.
\newblock Voter model on signed social networks.
\newblock {\em Internet Mathematics}, 11(2):93--133, 2015.
\newblock \href {http://dx.doi.org/10.1080/15427951.2013.862884}
  {\path{doi:10.1080/15427951.2013.862884}}.

\bibitem{XL-QJ-LW:18}
X.~Lin, Q.~Jiao, and L.~Wang.
\newblock Opinion propagation over signed networks: {M}odels and convergence
  analysis.
\newblock {\em IEEE Transactions on Automatic Control}, 64(8):3431--3438, 2019.
\newblock \href {http://dx.doi.org/10.1109/TAC.2018.2879568}
  {\path{doi:10.1109/TAC.2018.2879568}}.

\bibitem{JL-XC-TB-MAB:17}
J.~Liu, X.~Chen, T.~Ba\c{s}ar, and M.-A. Belabbas.
\newblock Exponential convergence of the discrete- and continuous-time
  {Altafini} models.
\newblock {\em IEEE Transactions on Automatic Control}, 62:6168--6182, 2017.
\newblock \href {http://dx.doi.org/10.1109/TAC.2017.2700523}
  {\path{doi:10.1109/TAC.2017.2700523}}.

\bibitem{ZM-GS-KHJ-MC-YH:16}
Z.~Meng, G.~Shi, K.~H. Johansson, M.~Cao, and Y.~Hong.
\newblock Behaviors of networks with antagonistic interactions and switching
  topologies.
\newblock {\em Automatica}, 73:110--116, 2016.
\newblock \href {http://dx.doi.org/10.1016/j.automatica.2016.06.022}
  {\path{doi:10.1016/j.automatica.2016.06.022}}.

\bibitem{AVP-RT:17}
A.~V. Proskurnikov and R.~Tempo.
\newblock A tutorial on modeling and analysis of dynamic social networks. {Part
  I}.
\newblock {\em Annual Reviews in Control}, 43:65--79, 2017.
\newblock \href {http://dx.doi.org/10.1016/j.arcontrol.2017.03.002}
  {\path{doi:10.1016/j.arcontrol.2017.03.002}}.

\bibitem{GS-CA-JB:19}
G.~Shi, C.~Altafini, and J.~Baras.
\newblock Dynamics over signed networks.
\newblock {\em SIAM Review}, 61(2):229--257, 2019.
\newblock \href {http://dx.doi.org/10.1137/17M1134172}
  {\path{doi:10.1137/17M1134172}}.

\end{thebibliography}

\end{document}